%% file: main.tex
\documentclass[11pt]{article}
\usepackage[margin=1in]{geometry}

\usepackage[english]{babel}
\usepackage[utf8]{inputenc}

\usepackage[T1]{fontenc}
\usepackage{kpfonts}
\usepackage{libertine}
\usepackage[scaled=0.85]{beramono}

\usepackage[table]{xcolor} 

\usepackage{graphicx}
\usepackage[colorinlistoftodos]{todonotes}
\usepackage{hyperref}
\usepackage{amsthm}
\usepackage{amsmath}
\usepackage{amssymb}
\usepackage{amsfonts}
\usepackage{mathrsfs}
\usepackage{mathtools}
\usepackage{verbatim}
\usepackage{footnote}
\usepackage{algorithm}
\usepackage[noend]{algpseudocode}
\usepackage{lineno}
\usepackage{caption}
\usepackage{tikzsymbols}
\hypersetup{colorlinks={true},linkcolor={blue},citecolor=cyan}
\usepackage{tikz}
\usetikzlibrary{arrows.meta}
\usepackage{boxedminipage2e}

\usepackage{eqparbox,array}

\usepackage{adjustbox,booktabs,multirow,tabularx}

\newtheorem{definition}{Definition}
\newtheorem{lemma}{Lemma}
\newtheorem{theorem}{Theorem}
\newtheorem{observation}{Observation}
\DeclareMathOperator*{\argmax}{arg\,max}

\newcommand{\ra}[1]{\renewcommand{\arraystretch}{#1}}
\input{header.tex}
\begin{document}
\title{How to Solve Fair $k$-Center in Massive Data Models} \author{
  Ashish Chiplunkar \\
  Indian Institute of Technology Delhi\\
  \texttt{ashishc@iitd.ac.in}
  \and
  Sagar Kale \\
  University of Vienna\\
  \texttt{sagar.kale@univie.ac.at}
  \and
  Sivaramakrishnan Natarajan Ramamoorthy\\
  University of Washington\\
  \texttt{sivanr@cs.washington.edu}
 }
\date{}

\maketitle

\input{abstract.tex}
\input{introduction.tex}
\input{preliminaries.tex}
\section{Algorithms}
\input{basic_procs.tex}
\input{two_pass.tex}
\input{distributed.tex}
\input{guesses.tex}
\input{lb.tex}
\input{experiments1.tex}
\input{experiments2.tex}

\input{conclusion.tex}

\bibliographystyle{alpha}
\bibliography{references}

\appendix
\input{supplementary.tex}

\end{document}

%% file: header.tex
\DeclareMathOperator{\opt}{OPT}

\renewcommand{\ge}{\geqslant}
\renewcommand{\le}{\leqslant}
\renewcommand{\geq}{\geqslant}
\renewcommand{\leq}{\leqslant}

\newcommand{\eps}{\varepsilon}

\newcommand{\lobebe}{\log(1/\eps)/\eps}


%% file: abstract.tex
\begin{abstract}
  Fueled by massive data, important decision making is being automated with the help of algorithms, therefore, fairness in algorithms has become an especially important research topic. In this work, we design new streaming and distributed algorithms for the fair $k$-center problem that models fair data summarization.  The streaming and distributed models of computation have an attractive feature of being able to handle massive data sets that do not fit into main memory. Our main contributions are: (a) the first distributed algorithm; which has provably constant approximation ratio and is extremely parallelizable, and (b) a two-pass streaming algorithm with a provable approximation guarantee matching the best known algorithm (which is not a streaming algorithm). Our algorithms have the advantages of being easy to implement in practice, being fast with linear running times, having very small working memory and communication, and outperforming existing algorithms on several real and synthetic data sets. To complement our distributed algorithm, we also give a hardness result for natural distributed algorithms, which holds for even the special case of $k$-center.  
\end{abstract}


%% file: introduction.tex
\section{Introduction}
Data summarization is a central problem in the area of machine learning, where
we want to compute a small summary of the data.  For example, if the input data
is enormous, we do not want to run our machine learning algorithm on the whole
input but on a small \emph{representative} subset.  How we select such a
representative summary is quite important.  It is well known that if the input is biased, then the
machine learning algorithms trained on this data will exhibit the same bias.
This is a classic example of \emph{selection bias} but as exhibited by
algorithms themselves.  Currently used algorithms for data summarization have
been shown to be biased with respect to attributes such as gender, race, and age
(see, e.g., \cite{kay_etal15}), and this motivates the fair data summarization
problem.  Recently, the fair $k$-center
problem was shown to be useful in computing fair summary~\cite{kleindessner_etal19a}.
In this paper, we continue the study of
fair $k$-center and add to the series of works on fairness in machine learning
algorithms.  Our main results are streaming and distributed algorithms for fair
$k$-center.  These models are extremely suitable for handling massive datasets.
The fact that data summarization problem arises when the input is huge makes our
work all the more relevant!

Suppose the input is a set of real vectors with a gender attribute and you want to compute a
summary of $k$ data points such that both\footnote{sincere apologies
  to the people who identify with neither} genders are represented equally.  Say
we are given a summary $S$.  The cost we pay for not including a point in $S$ is
its Euclidean distance from $S$.  Then the cost of $S$ is the
largest cost of a point.  We want to compute a summary with minimum cost that is
also fair, i.e., contains $k/2$ women and $k/2$ men.  In one sentence, we want
to compute a fair summary such that the point that is farthest from this summary
is not too far.  Fair $k$-center models this task: let the number of points
in the input be $n$, the number of \emph{groups} be $m$, target summary size be
$k$, and we want to select a summary $S$ such that $S$ contains $k_j$ points belonging to Group~$j$, where $\sum_j k_j = k$.
And we want to minimize $\max_x d(x, S) = \max_x \min_{x'\in S} d(x, x')$, where $d$ denotes the distance function.  Note that each point
belongs to exactly one of the $m$ groups; for the case of gender, $m = 2$.

We call the special case 
where
  $m = 1$ and $k_1 = k$ as just $k$-center throughout this paper.
For $k$-center, there are simple greedy algorithms with an approximation ratio of $2$~\cite{gonzalez85,hochbaum85}, and
 getting better than $2$-approximation is
NP-hard~\cite{hsu79}.  The NP-hardness result also applies to the more general fair $k$-center.
The best algorithm known for fair $k$-center is a
$3$-approximation algorithm that runs in time $O(n^2\log n)$~\cite{chen_etal16}.
A linear-time algorithm with approximation guarantee of $O(2^m)$, which is
constant if $m$ is, was given recently~\cite{kleindessner_etal19a}.  Both of
these algorithms work only in the traditional random access machine model, which
is suitable only if the input is small enough to fit into fast memory.  We give
a two-pass streaming algorithm that achieves the approximation ratio arbitrarily
close to $3$.  In the streaming setting, input is thought to arrive one point at
a time, and the algorithm has to process the input quickly, using minimum amount
of working memory---ideally linear in the size of a feasible solution, which is
$k$ for fair $k$-center.  Our algorithm processes each incoming input point in
$O(k)$ time and uses space $O(km)$, which is $O(k)$ if the number of groups $m$
is very small.  This improves the space usage of the existing streaming
algorithm~\cite{kale19} almost quadratically, from $O(k^2)$, while also
matching the best approximation ratio achieved by Chen et al.  We also give the first
distributed, constant approximation algorithm where the input is divided among
multiple processors, each of which performs one round of computation and sends a message
of size $O(km)$ to a central processor, which then computes the final solution.  Both
rounds of computation are linear time.  All the approximation, communication,
space usage, and running-time guarantees are provable.  To complement our
distributed algorithm, we prove that any distributed algorithm, even randomized, that works by
each processor sending a subset of its input to a central processor which outputs
the solution, needs to essentially communicate the whole input to achieve an approximation ratio of better than
$4$.  This, in fact, applies for the special case of $k$-center showing that
known $4$-approximation algorithm~\cite{malkomes1_etal15} for $k$-center is
optimal.

We perform experiments on real and synthetic datasets and show that our
algorithms are as fast as the linear-time algorithm of Kleindessner et al.,
while achieving improved approximation ratio, which matches that of Chen et al.
Note that this comparison is possible only for small datasets, since those
algorithms do not work either in streaming or in distributed setting.  We also
run our algorithms on a really large synthetic dataset of size 100GB, and show that
their running time is only one order of magnitude more than the time taken to just read the input dataset
from secondary memory.

As a further contribution, we give faster implementations of existing
algorithms---those of Kale and Chen et al.

\subsection*{Related work}
Chen et al.\ gave the first polynomial-time algorithm that achieves $3$-approximation.  Kale achieves almost the same ratio using just two passes and also gives a one-pass $(17 + \eps)$-approximation algorithm, both using
$O(k^2)$ space.

One way that is incomparable to ours is to compute a fair summary is using a determinantal
measure of diversity~\cite{celis_etal18a}.  Fair clustering has been studied
under another notion of fairness, where each cluster must be balanced with
respect to all the groups (no over-or-under-representation of any
group)~\cite{chierichetti_etal17}, and this line of work also has received a lot
of attention in a short span of
time~\cite{BeraCFN19,AhmadianE0M19,bandyapadhyay_et_al19,schmidt_etal20,jia_etal20}.

The $k$-median clustering problem with fairness constraints was first considered
by~\cite{hajiaghayi_etal10} and with more general matroid constraints was
studied by~\cite{krishnaswamy_etal11}.  The work of Chen et al.\ and Kale also
actually applies for matroid constraints.

There has been a lot of work done on fairness, and we
refer the reader to overviews by~\cite{kleindessner_etal19a,celis_etal18a}.


%% file: preliminaries.tex
\section{Preliminaries}
\label{sec:prelim}
The input to fair $k$-center is a set $X$ of $n$ points in a metric space given by a distance function $d$. We denote this metric space by $(X,d)$. Each point belongs to one of $m$ groups, say $\{1,\ldots,m\}$. Let $g:X\longrightarrow\{1,\ldots,m\}$ denote this group assignment function. Further, for each group $j$, we are given a capacity $k_j$. Let $k=\sum_{j=1}^mk_j$. We call a subset $S\subseteq X$ \textit{feasible} if for every $j$, the set $S$ contains at most $k_j$ points from group $j$. The goal is to compute a feasible set of centers that (approximately) minimizes the clustering cost, formally defined as follows.

\begin{definition}\label{def_cost}
Let $A,B\subseteq X$, then the \textit{clustering cost} of $A$ for $B$ is defined as $\max_{b\in B}\min_{a\in A}d(a,b)$.
\end{definition}

Note here that we allow $A$ to not be a subset of $B$. The following lemmas follow easily from the fact that the distance function $d$ satisfies the triangle inequality.

\begin{lemma}\label{lem_cost_triangle}
Let $A,B,C\subseteq X$. The clustering cost of $A$ for $C$ is at most the clustering cost of $A$ for $B$ plus the clustering cost of $B$ for $C$.
\end{lemma}

\begin{lemma}\label{lem_pigeonhole}
Suppose for a set $T$ of points there exists a set $S$ of $k$ centers, not necessarily a subset of $T$, whose clustering cost for $T$ is at most $\rho$. If $P\subseteq T$ is a set of points separated pairwise by distance more than $2\rho$, then $|P|\leq k$.
\end{lemma}

\begin{proof}
If $|P|>k$ then some two points in $P$ must share one of the $k$ centers, and must therefore be both within distance $\rho$ from that common center. Then by the triangle inequality, they cannot be separated by distance more than $2\rho$.
\end{proof}

We denote by $S^*$ a feasible set which has the minimum clustering cost for $X$, and by $\text{OPT}$ the minimum clustering cost. We assume that our algorithms have access to an estimate $\tau$ of $\text{OPT}$. When $\tau$ is at least $\text{OPT}$, our algorithms compute a solution of cost at most $\alpha\tau$ for a constant $\alpha$. Thus, when $\tau\in[\text{OPT},(1+\varepsilon)\text{OPT}]$, our algorithms compute a $(1+\varepsilon)\alpha$-approximate solution. In Section~\ref{subsec:guesses} we describe how to efficiently compute such a $\tau$.

%% file: basic_procs.tex

Before stating algorithms, we describe some elementary procedures which will be used as subroutines in our algorithms.

\textbf{\texttt{getPivots}}$(T,d,r)$ takes as input a set $T$ of points with distance function $d$ and a radius $r$. Starting with $P=\emptyset$, it performs a single pass over $T$. Whenever it finds a point $q$ which is not within distance $r$ from any point in $P$, it adds $q$ to $P$. Finally, it returns $P$. Thus, $P$ is a maximal subset of $T$ of points separated pairwise by distance more than $r$. We call points in $P$ \textit{pivots}. By Lemma~\ref{lem_pigeonhole}, if there is a set of $k$ points whose clustering cost for $T$ is at most $r/2$, then $|P|\leq k$. Moreover, due to maximality of $P$, its clustering cost for $T$ is at most $r$. Note that \texttt{getPivots}$()$ runs in time $O(|P|\cdot|T|)$.

\textbf{\texttt{getReps}}$(T,d,g,P,r)$ takes as input a set $T$ of points with distance function $d$, a group assignment function $g$, a subset $P\subseteq T$, and a radius $r$. For each $p\in P$, initializing $N(p)=\{p\}$, it includes in $N(p)$ one point, from each group, which is within distance $r$ from $p$ whenever such a point exists. Note that this is done while performing a single pass over $T$. This procedure runs in time $O(|P|\cdot|T|)$.

Informally, if $P$ is a good but infeasible set of centers, then \texttt{getReps}$()$ finds representatives $N(p)$ of the groups in the vicinity of each $p\in P$. This, while increasing the clustering cost by at most $r$, gives us enough flexibility to construct a feasible set of centers. The procedure \texttt{HittingSet}$()$ that we describe next finds a feasible set from a collection of sets of representatives.

\textbf{\texttt{HittingSet}}$(\mathcal{N},g,\overline{k})$ takes as input a collection $\mathcal{N}=\{N_1,\ldots,N_K\}$ of pairwise disjoint sets of points, a group assignment function $g$, and a vector $\overline{k}=(k_1,\ldots,k_m)$ of capacities of the $m$ groups. It returns a feasible set $S$ intersecting as many $N_i$'s as possible. This reduces to finding a maximum cardinality matching in an appropriately constructed bipartite graph. It is important to note that this procedure does the post-processing: it doesn't make any pass over the input stream of points. This procedure runs in time $O(K^2\cdot\max_i|N_i|)$.

For interested readers, the pseudocodes of these procedures, an explanation of \texttt{HittingSet}$()$, and the proof of its running time appear in Appendix~\ref{app:alg}.


%% file: two_pass.tex
\subsection{A Two-Pass Algorithm}

\begin{algorithm}[tb]
\caption{Two-pass algorithm}
\label{alg_2pass}
\begin{algorithmic}
\State {\bfseries Input}: Metric space $(X,d)$, group assignment function $g$, capacity vector $\overline{k}$.
\State /* Pass 1: Compute pivots. */
\State $P$ $\gets$ \texttt{getPivots}$(X,d,2\tau)$.
\State /* Pass 2: Compute representatives. */
\State $\{N(q):q\in P\}$ $\gets$ \texttt{getReps}$(X,d,g,P,\tau)$.
\State /* Compute solution. */
\State $S$ $\gets$ \texttt{HittingSet}$(\{N(q):q\in P\},g,\overline{k})$.
\State {\bfseries Output} $S$.
\end{algorithmic}
\end{algorithm}

Recall that $\tau$ is an upper bound on the minimum clustering cost. Our two-pass algorithm given by Algorithm~\ref{alg_2pass} consists of three steps. First, the algorithm constructs a maximal subset $P\subseteq X$ of pivots separated pairwise by distance more than $2\tau$ by executing one pass on the stream of points. In another pass, the algorithm computes a representative set $N(q)$ of each pivot $q\in P$. Points in the representative set of a pivot are within distance $\tau$ from the pivot. Due to the separation of $2\tau$ between the pivots, these representative sets are pairwise disjoint. Finally, a feasible set $S$ intersecting as many $N(q)$'s as possible is found and returned. (It will soon be clear that $S$ intersects all the $N(q)$'s.)

The algorithm needs working space only to store the pivots and their representative sets. By substituting $S=S^*$ in Lemma~\ref{lem_pigeonhole}, the number of pivots is at most $k$, that is, $|P|\leq k$. Since $N(q)$ contains at most one point from any group, it has at most $m-1$ points other than $q$. Thus,

\begin{observation}
The two-pass algorithm needs just enough working space to store $km$ points.
\end{observation}

The calls to \texttt{getPivots} and \texttt{getReps} both take time $O(|P|\cdot|X|)=O(kn)$, with $O(|P|)=O(k)$ update time per point. The call to \texttt{HittingSet} takes time $O(|P|^2\cdot\max_{q\in P}|N(q)|)=O(mk^2)$. Thus,

\begin{observation}
The two-pass algorithm runs in time $O(kn+mk^2)$, which is $O(kn)$ when $m$, the number of groups, is constant.
\end{observation}

We now prove the approximation guarantee.

\begin{theorem}
The two-pass algorithm returns a feasible set whose clustering cost is at most $3\tau$. This is a $3(1+\varepsilon)$-approximation when $\tau\in[\text{OPT},(1+\varepsilon)\text{OPT}]$.
\end{theorem}

\begin{proof}
Recall that $S^*$ is a feasible set having clustering cost at most $\tau$. For each $q\in P$ let $c_q\in S^*$ denote a point such that $d(q,c_q)\leq\tau$. Since the points in $P$ are separated by distance more than $2\tau$, the points $c_q$ are all distinct. Recall that $N(q)$, the output of \texttt{getReps}$()$, contains one point from every group which has a point within distance $\tau$ from $q$. Therefore, $N(q)$ contains a point, say $b_q$, from the same group as $c_q$ such that $d(q,b_q)\leq\tau$. Consider the set $B=\{b_q:q\in P\}$. This set intersects $N(q)$ for each $q$. Furthermore, $B$ contains exactly as many points from any group as $\{c_q:q\in P\}\subseteq S^*$, and therefore, $B$ is feasible. Thus, there exists a feasible set, namely $B$, intersecting all the pairwise disjoint $N(q)$'s. Recall that $S$, the output of \texttt{HittingSet}$()$, is a feasible set intersecting as many $N(q)$'s as possible. Thus, $S$ also intersects all the $N(q)$'s.

Now, the clustering cost of $S$ for $P$ is at most $\tau$, because $S$ intersects $N(q)$ for each $q\in P$. The clustering cost of $P$ for $X$ is at most $2\tau$ by the maximality of the set returned by \texttt{getPivots}$()$. These facts and Lemma~\ref{lem_cost_triangle} together imply that the clustering cost of $S$, the output of the algorithm, for $X$ is at most $3\tau$.
\end{proof}

%% file: distributed.tex
\subsection{A Distributed Algorithm}

In the distributed model of computation, the set $X$ of points to be clustered is distributed equally among $\ell$ processors. Each processor is allowed a restricted access to the metric $d$: it may compute the distance between only its own points. Each processor performs some computation on its set of points and sends a summary of small size to a coordinator. From the summaries, the coordinator then computes a feasible set $S$ of points which covers all the $n$ points in $X$ within a small radius. Let $X_i$ denote the set of points distributed to processor $i$. 

\begin{algorithm}[tb]
\caption{Summary computation by the $i$'th processor}
\label{alg_slave}
\begin{algorithmic}
\State {\bfseries Input:} Set $X_i$, metric $d$ restricted to $X_i$, group assignment function $g$ restricted to $X_i$.
\State /* Compute local pivots. */
\State $p^i_1$ $\gets$ an arbitrary point in $X_i$.
\For{$j=2$ {\bfseries to} $k+1$}
    \State $p^i_j\gets\argmax_{p\in X_i}\min_{j':1\leq j'<j}d(p,p^i_j)$.
\EndFor
\State $P_i\gets\{p^i_1,\ldots,p^i_k\}$.
\State $r_i\gets\min_{j':1\leq j'\leq k}d(p^i_{k+1},p^i_j)/2$.
\State /* Compute local representative sets. */
\State $\{L(p):p\in P_i\}$ $\gets$ \texttt{getReps}$(X_i,d,g,P_i,2r_i)$.
\State $L_i\gets\bigcup_{p\in P_i}L(p)$.
\State /* Send message to coordinator. */
\State Send $(P_i,L_i)$ to the coordinator.
\end{algorithmic}
\end{algorithm}

The algorithm executed by each processor $i$ is given by Algorithm~\ref{alg_slave}, which consists of two main steps. In the first step, the processor uses Gonzalez's farthest point heuristic to find $k+1$ points. The first $k$ of those constitute the set $P_i$, which we will call the set of \textit{local pivots}. The point $p_{k+1}$ is the farthest point from the set of local pivots, and it is at a distance $2r_i$ from the set of local pivots. Thus, every point $X_i$ is within distance $2r_i$ from the set of pivots. This means,

\begin{observation}\label{obs_Pi_Xi}
The clustering cost of $P_i$ for $X_i$ is $2r_i$.
\end{observation}

In the second step, for each local pivot $p\in P_i$, the processor computes a set $L(p)$ of local representatives in the vicinity of $p$. Finally, the set $P_i$ of local pivots and the union $L_i=\bigcup_{p\in P_i}L(p)$ of local representative sets is sent to the coordinator. Since $L(p)$ contains at most one point from any group, it has at most $m-1$ points other than $p$. Since $|P_i|=k$ we have the following observation.

\begin{observation}\label{obs_msg}
Each processor sends at most $km$ points to the coordinator.
\end{observation}

Moreover, the separation between the local pivots is bounded as follows.

\begin{lemma}\label{lem_tau}
For every processor $i$, we have $r_i\leq\text{OPT}\leq\tau$.
\end{lemma}

\begin{proof}
Suppose $r_i>\tau$. Then $\{p^i_1,\ldots,p^i_{k+1}\}\subseteq X_i$ is a set of $k+1$ points separated pairwise by distance more than $2\tau$. But $S^*$ is a set of at most $k$ points whose clustering cost for $X_i$ is $\text{OPT}\leq\tau$. This contradicts Lemma~\ref{lem_pigeonhole}.
\end{proof}

Observation~\ref{obs_Pi_Xi} allows us to define a covering function $\text{cov}$ from $X$, the input set of points, to $\bigcup_{i=1}^{\ell}P_i$, the set of local pivots, as follows.

\begin{definition}\label{def_cov}
Let $p$ be an arbitrary point in $X$. Suppose $p$ is processed by processor $i$, that is, $p\in X_i$. Then $\text{cov}(p)$ is an arbitrary local pivot in $P_i$ within distance $2r_i$ from $p$.
\end{definition}

Since the processors send only a small number of points to the coordinator, it is very well possible that the optimal set $S^*$ of centers is lost in this process. In the next lemma, we claim that the set of points received by the coordinator contains a good and feasible set of centers nevertheless.

\begin{lemma}\label{lem_5tau}
The set $L=\bigcup_{i=1}^{\ell}L_i$ contains a feasible set, say $B$, whose clustering cost for $\bigcup_{i=1}^{\ell}P_i$ is at most $5\tau$.
\end{lemma}

\begin{proof}
Consider any $c\in S^*$, and suppose it is processed by processor $i$. Then $d(c,\text{cov}(c))\leq2r_i$ by Definition~\ref{def_cov}. Recall that $L(\text{cov}(c))$, the output of \texttt{getReps}$()$, contains one point from every group which has a point within distance $2r_i$ from $\text{cov}(c)$. Therefore, $L(\text{cov(c)})\subseteq L_i$ contains some point, say $c'$, from the same group as $c$ (possibly $c$ itself), such that $d(c',\text{cov}(c))\leq2r_i$. Then $d(c,c')\leq4r_i\leq4\tau$ by the triangle inequality and Lemma~\ref{lem_tau}. Let $B=\{c':c\in S^*\}$. Clearly, $B\subseteq\bigcup_{i=1}^{\ell}L_i$. Since $B$ has exactly as many points from any group as $S^*$, $B$ is feasible. The clustering cost of $B$ for $S^*$ is at most $4\tau$. The clustering cost of $S^*$ for $\bigcup_{i=1}^{\ell}P_i$ is at most $\tau$, because $\bigcup_{i=1}^{\ell}P_i\subseteq X$. By Lemma~\ref{lem_cost_triangle}, the clustering cost of $B$ for $\bigcup_{i=1}^{\ell}P_i$ is at most $5\tau$, as required.
\end{proof}

\begin{algorithm}[tb]
\caption{Coordinator's algorithm}
\label{alg_master}
\begin{algorithmic}
\State $X'\gets\emptyset$, $L\gets\emptyset$.
\State /* Receive messages from processors. */
\For{$i=1$ {\bfseries to} $\ell$}
    \State Receive $(P_i,L_i)$ from processor $i$.
    \State $X'\gets X'\cup P_i$, $L\gets L\cup L_i$.
\EndFor
\State /* Coordinator now has access to $d$ and $g$ restricted to $X'\cup L$, and capacity vector $\overline{k}=(k_1,\ldots,k_m)$. */
\State /* Compute global pivots. */
\State $P$ $\gets$ \texttt{getPivots}$(X',d,10\tau)$.
\State /* Compute global representative sets. */
\State $\{N(q):q\in P\}$ $\gets$ \texttt{getReps}$(L,d,g,P,5\tau)$.
\State /* Compute solution. */
\State $S$ $\gets$ \texttt{HittingSet}$(\{N(q):q\in P\},g,\overline{k})$.
\State {\bfseries Output} $S$.
\end{algorithmic}
\end{algorithm}

The algorithm executed by the coordinator is given by Algorithm~\ref{alg_master}. The coordinator constructs a maximal subset $P$ of the set of pivots $X'=\bigcup_{i=1}^{\ell}P_i$ returned by the processors such that points in $P$ are pairwise separated by distance more than $10\tau$. $P$ is called the set of global pivots. For each global pivot $q\in P$, the coordinator computes a set $N(q)\subseteq L=\bigcup_{i=1}^{\ell}L_i$ of its global representatives, all of which are within distance $5\tau$ from $q$. Due to the separation between points in $P$, the sets $N(q)$ are pairwise disjoint. Finally, a feasible set $S$ intersecting as many $N(q)$'s as possible is found and returned. (As before, it will be clear that $S$ intersects all the $N(q)$'s.)

\begin{theorem}
The coordinator returns a feasible set whose clustering cost is at most $17\tau$. This is a $17(1+\varepsilon)$-approximation when $\tau\in[\text{OPT},(1+\varepsilon)\text{OPT}]$.
\end{theorem}

\begin{proof}
By Lemma~\ref{lem_5tau}, $L$ contains a feasible set, say $B$, whose clustering cost for $X'$ is at most $5\tau$. For each $q\in P\subseteq X'$, let $b_q$ denote a point in $B$ that is within distance $5\tau$ from $q$. Since the points in $X'$ are separated pairwise by distance more than $10\tau$, $b_q$'s are all distinct. 
By the property of \texttt{getReps}$()$, the set $N(q)$ returned by it contains a point, say $b'_q$, from the same group as $b_q$.
Let $B'=\{b'_q:q\in P\}$. This set $B'$ intersects $N(q)$ for each $q\in P$. Since $b'_q$ and $b_q$ are from the same group and $b_q$'s are all distinct, $B'$ contains at most as many points from any group as $B$ does. Since $B$ is feasible, so is $B'$. To summarize, there exists a feasible set, namely $B'$, intersecting all the $N(q)$'s.
Recall that $S$, the output of \texttt{HittingSet}$()$, is a feasible set intersecting as many $N(q)$'s as possible. Thus, $S$ also intersects all the $N(q)$'s.

Now, the clustering cost of $S$ for $P$ is at most $5\tau$, because $S$ intersects $N(q)$ for each $q\in P$. The clustering cost of $P$ for $X'$ is at most $10\tau$ by the maximality of the set returned by \texttt{getPivots}$()$. The clustering cost of $X'=\bigcup_{i=1}^{\ell}P_i$ for $X=\bigcup_iX_i$ is at most $2\tau$ because the clustering cost of each $P_i$ for $X_i$ is at most $2r_i\leq2\tau$. These facts and Lemma~\ref{lem_cost_triangle} together imply that the clustering cost of $S$, the output of the coordinator, for $X$ is at most $17\tau$.
\end{proof}

We note here that even though our distributed algorithm has the same approximation guarantee as Kale's one-pass algorithm,
 it is inherently a different algorithm.  Ours is extremely parallel whereas Kale's is extremely sequential.  We now prove a bound on the running time.

\begin{theorem}
The running time of the distributed algorithm is $O(kn/\ell+mk^2\ell)$.  By an appropriate choice of $\ell$, the number of processors, this can be made $O(m^{1/2}k^{3/2}n^{1/2})$.
\end{theorem}

\begin{proof}
For each processor $i$, computing local pivots as well as the call to \texttt{getReps}$()$ takes $O(|P_i|\cdot|X_i|)=O(kn/\ell)$ time each. For the coordinator, the separation between the global pivots and Lemma~\ref{lem_pigeonhole} together enforce $|P|\leq k$. Observation~\ref{obs_msg} implies $|L|\leq m\cdot\max_i|L_i|\leq mk\ell$. Therefore, \texttt{getPivots}$()$ takes time $O(|P|\cdot|X'|)=O(k^2\ell)$ and \texttt{getReps}$()$ takes time $O(|P|\cdot|L|)=O(mk^2\ell)$. The call to \texttt{HittingSet}$()$ takes time $O(k^2\max_q|N(q)|)=O(mk^2)$, thus limiting the coordinator's running time to $O(mk^2\ell)$. Choosing $\ell=\Theta(\sqrt{n/(mk)})$ minimizes the total running time to $O(m^{1/2}k^{3/2}n^{1/2})$.
\end{proof}

%% file: guesses.tex
\subsection{Handling the Guesses}
\label{subsec:guesses}
Given an arbitrarily small parameter $\varepsilon$, a lower bound $L
\le \opt$, and an upper bound $U \ge \opt$, we run our algorithms for guess
$\tau \in \{L, L(1+\eps), L(1+\eps)^2, \ldots, U\}$, which means at most
$\log_{1+\eps}(U/L)$ guesses.  Call this method of guesses as geometric guessing starting at $L$ until $U$.  For the $\tau \in [\opt, \opt (1 + \eps)]$, our
algorithms will compute a solution successfully.

In the distributed algorithm, by Lemma~\ref{lem_tau}, for each
processor, $r_i \le \opt$.  Therefore,
$\max_i r_i \le \opt$.  We then run Algorithm~\ref{alg_master} with geometric guessing starting at $\max_i r_i$ until it
successfully finds a solution.

For the two-pass algorithm, let $P$ be the set of first
$k + 1$ points; then $L = \min_{x_1, x_2 \in P} d(x_1, x_2)/2$ is a
lower bound (call this the simple lower bound).  Note that no passes need to be spent to compute the simple lower bound.  We also need an upper bound
$U \ge \opt$.  One can compute an arbitrary solution and its cost---which will
be an upper bound---by spending two more passes (call this the simple upper bound).  This results in a four-pass
algorithm.  To obtain a truly two pass algorithm and space usage
$O(km \lobebe)$, one can use Guha's trick~\cite{guha09}, which is essentially
starting $O(\lobebe)$ guesses and if a run with guess $\tau$ fails, then continuing the run with guess $\tau/\eps$ and treating the old summary as the initial
stream for this guess; see also \cite{kale19} for details.  But obtaining and using an upper bound is convenient and easy to implement in practice.


%% file: lb.tex
\section{Distributed $k$-Center Lower Bound}
\label{sec:lb}
Malkomes et al.~\cite{malkomes1_etal15} generalized the greedy algorithm~\cite{gonzalez85} to obtain a $4$-approximation algorithm for the $k$-center problem in the distributed setting.
Here we prove a lower bound for the $3$-center problem with $9$ processors for a special class of distributed algorithms: 
If each processor communicates less than a constant fraction of their input points, then with a constant probability, the output of the coordinator will be no better than a $4$-approximation to the optimum.
\begin{figure}[h]
    \centering{\includegraphics[scale=0.6]{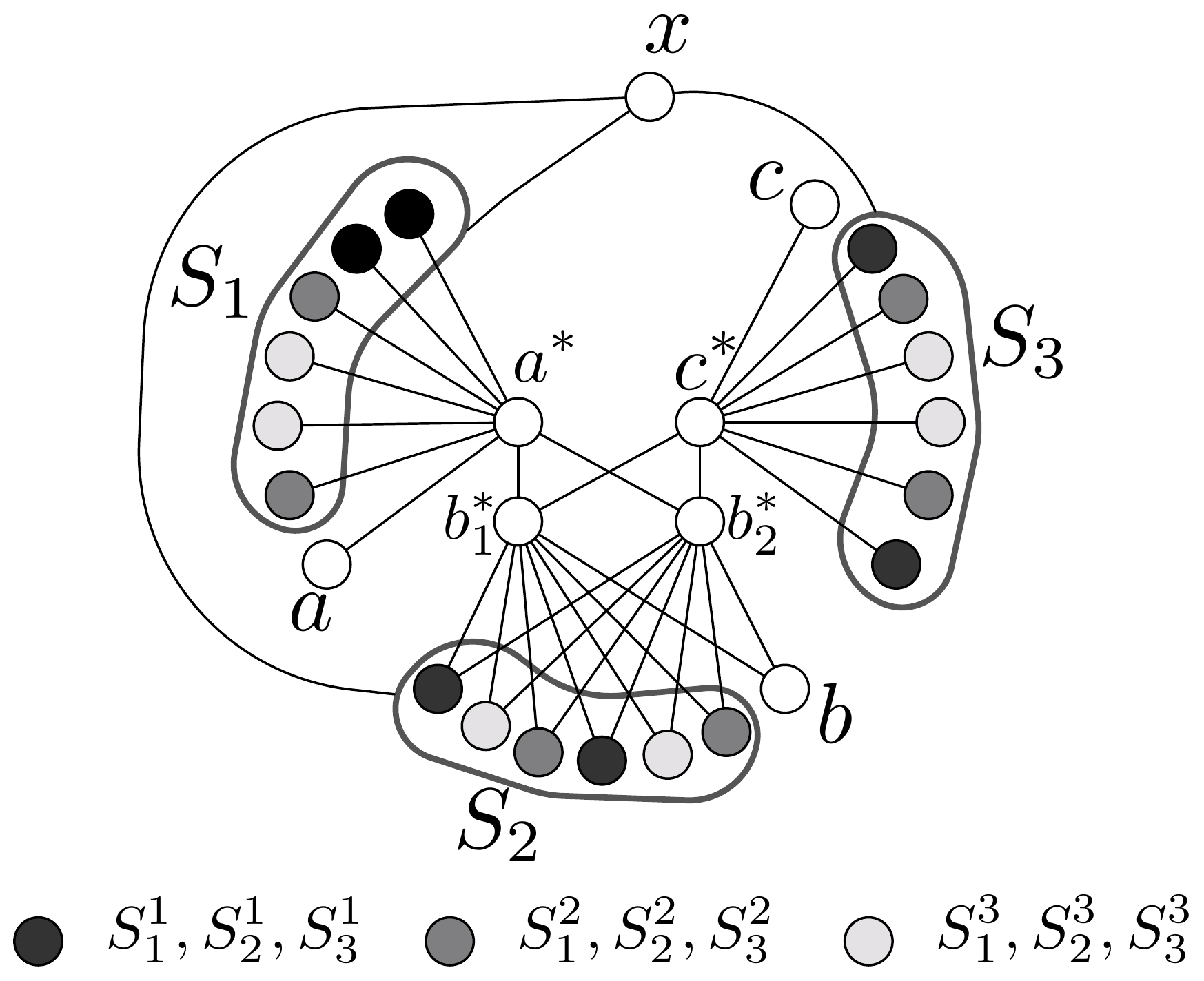}}
    \caption{\footnotesize The underlying metric for $n'=2$}
    \label{metric}
\end{figure}
Figure~\ref{metric} shows a graph metric with $9n'+7$ points for which lower bound holds, where the point $x$ is not a part of the metric but is only used to define the distances.
Note that $|S_1|=|S_2|=|S_3| = 3n'$ and $x$ is at distance of $1$ from each point in $S_1 \cup S_2 \cup S_3$. 

For $i\in \{1,2,3\}$, let $S_i^1,S_i^2,S_i^3$ denote an arbitrary equipartition of 
$S_i$. 
There are $9$ processors, whose inputs are given by $Y_1^j = \{b_1^*,b_2^*,a\} \cup S_1^j$, $Y_2^j = \{a^*,c^*,b\} \cup S_2^j$ and $Y_3^j = \{b_1^*,b_2^*,c\} \cup S_3^j$, for $j\in \{1,2,3\}$.
The goal is to solve the $3$-center problem on the union of their inputs.
(Observe that the optimum solution is $\{a^*,c^*,b^*_1\}$ with distance $1$.)
Each processor is allowed to send a subset of their input points to the coordinator, who outputs three of the received points.
For this class of algorithms, we show that if each processor communicates less than $(n' + 3)/54$ points, then the output of the coordinator is no better than a $4$-approximation to the optimum with probability at least $1/84$.
Using standard amplification arguments, we can generate a metric instance for the ($3\alpha$)-center problem on which with probability at least $1-\varepsilon$, the algorithm outputs no better than $4$-approximation ($\alpha \approx \log(1/\varepsilon)$).

We first discuss the intuition behind the proof.
The key observation is that all points in each $Y_i^j$ are pairwise equidistant.
Therefore, sending a uniformly random subset of the inputs is the best strategy for each processor.
Since each processor communicates only a small fraction of its input points, the probability that the coordinator receives any of the points in $\{a^*,b_1^*,b_2^*,c^*,a,b,c\}$ is negligible.
Conditioned on the coordinator not receiving these points, all the received points are a subset of $S_1\cup S_2 \cup S_3$.
As all points in $S_1\cup S_2 \cup S_3$ are pairwise equidistant, the best strategy for the coordinator is to output $3$ points at random.
Hence, with constant probability, all the points in the output belong to $S_1$ or all of them belong to $S_3$.
This being the case, the output has cost $4$, whereas the optimum cost is $1$.

\input{lb-proof.tex}


%% file: lb-proof.tex
\subsection{The Formal Proof}
We now present the formal details of the lower bound. 
For a natural number $n$, $[n]$ denotes the set $\{1,2,\ldots,n\}$.
\paragraph{The metric space $\mathcal{M}(n')$.} 
The point set of this metric space on $n = 9n'+7$ points is given by 
\[S := \{a^*,b_1^*,b_2^*,c^*,a,b,c\} \cup S_1 \cup S_2 \cup S_3,\]
where $|S_1| = |S_2| = |S_3| = 3n'$.
Let $C := \{a^*,b_1^*,b_2^*,c^*,a,b,c\}$.  We call the points in $C$ \emph{critical}.  Note that $S_1,S_2,S_3$ are pairwise disjoint and are also disjoint from $C$.  The metric $d:S\times S\longrightarrow\mathbb{R}$ is the shortest-path-length metric induced by the graph shown in Figure~\ref{metric} (where $x$ is not a point in $S$ but is only used to define the pairwise distances). The pairwise distances are given in Table~\ref{metric-pairwisedist}.
Note that if the table entry $i,j$ is indexed by sets, then the entry corresponds to the distance between distinct points in the sets. The following observation can be verified by a case-by-case analysis.

\begin{observation}\label{obs_lb_opt}
The sets $\{a^*,b_1^*,c^*\}$ and $\{a^*,b_2^*,c^*\}$ are the only optimum solutions of the $3$-center problem on $\mathcal{M}(n')$ and they have unit clustering cost. The clustering cost of any subset of $S_1$ is $4$ due to point $c$. Similarly, the clustering cost of any subset of $S_3$ is $4$ due to point $a$.
\end{observation}

\begin{table}[htp]
	\ra{1.5}
	\centering
	\adjustbox{max width = \textwidth}{
		\begin{tabular}{@{}c|ccccccccccc@{}}
			 & $a^*$ & $b_1^*$ & $b_2^*$ & $c^*$ & $a$ & $b$ & $c$  & $S_1$ & $S_2$ & $S_3$\\ 
			 \toprule
$a^*$& $0$ & $1$& $1$ &$2$ & $1$ & $2$ & $3$ & $1$ & $2$ & $3$\\
$b_1^*$ & $1$ & $0$ & $2$ &$1$ & $2$ & $1$ & $2$  & $2$ & $1$ & $2$ \\
$b_2^*$ & $1$ & $2$ &$0$ &$1$ & $2$ & $1$ & $2$ & $2$ & $1$ & $2$ \\
$c^*$ & $2$ & $1$ & $1$ &$0$ & $3$ & $2$ & $1$  & $3$ & $2$ & $1$ \\
$a$ & $1$ & $2$ & $2$ &$3$ & $0$ & $3$ & $4$  & $2$ & $3$ & $4$ \\
$b$ & $2$ & $1$ & $1$ &$2$ & $3$ & $0$ & $3$ & $3$ & $2$ & $3$ \\
$c$ & $3$ & $2$ & $2$ &$1$ & $4$ & $3$ & $0$  & $4$ & $3$ & $2$ \\
$S_1$ & $1$ & $2$ & $2$ &$3$ & $2$ & $3$ & $4$ & $2$ & $2$ & $2$ \\
$S_2$ & $2$ & $1$ & $1$ &$2$ & $3$ & $2$ & $3$ & $2$ & $2$ & $2$\\
$S_3$ & $3$ & $2$ & $2$ &$1$ & $4$ & $3$ & $2$  & $2$ & $2$ & $2$
	\end{tabular}}
	\caption{Pairwise Distances}
	\label{metric-pairwisedist}
\end{table}

\paragraph{Input Distribution $\mathcal{D}$ on the Processors' Inputs.}
For $i\in [3]$, let $S_i^1, S_i^2, S_i^3$ be an arbitrary equi-partition of $S_i$ (and therefore, $|S_i^j|=n'$ for all $i,j$). 
Define the sets 
$Y_1^j = \{b_1^*,b_2^*,a\} \cup S_1^j$, $Y_2^j = \{a^*,c^*,b\} \cup S_2^j$ and $Y_3^j = \{b_1^*,b_2^*,c\} \cup S_3^j$, for $j\in [3]$.
Observe that each $Y_i^j$ contains exactly $n'+3$ points separated pairwise by distance $2$, and moreover, three of the $n'+3$ points are critical.
We assign the sets $Y_i^j$ randomly to the nine processors after a random relabeling. Formally, we pick a uniformly random bijection $\pi:S\longrightarrow[n]$ as the relabeling and another uniformly random bijection $\Gamma:[3]\times[3]\longrightarrow[9]$, independent of $\pi$, as the assignment. We assign the set $\pi(Y_i^j)$ to processor $\Gamma(i,j)$ for every $i,j$. When a processor or the coordinator queries the distance between $p$ and $q$ where $p,q\in[n]$, it gets $d(\pi^{-1}(p),\pi^{-1}(q))$ as an answer.
Note that neither the processors nor the coordinator knows $\pi$ or $\Gamma$. Let the random variable $\mathcal{P}=(\mathcal{P}_1,\ldots,\mathcal{P}_9)$ denote the partition of the set of labels into a sequence of nine subsets induced by $\pi$ and $\Gamma$, where $\mathcal{P}_r$ is the set of labels of points assigned to processor $r$, that is, $\mathcal{P}_{\Gamma(i,j)}=\pi(Y_i^j)$.

\begin{lemma}\label{thm:lb}
Consider any deterministic distributed algorithm for the $9$ processor $3$-center problem on $\mathcal{M}(n')$ and input distribution $\mathcal{D}$, in which each processor communicates an $\ell$-sized subset of its input points, and the coordinator outputs $3$ of the received points.
If $\ell \leq (n'+3)/54$, then with probability at least $1/84$, the output is no better than a $4$-approximation.
\end{lemma}

Although the probability with which the coordinator fails to outputs a better-than-$4$-approximation is only $1/84$, it can be \emph{amplified} to $1-\varepsilon$, for any $\varepsilon > 0$.
We discuss the amplification result before presenting the proof of the above lemma.

\begin{lemma}\label{lem:amplifiedlb}
Let $\varepsilon>0$ and $c<1/486$ be arbitrary constants, and let
\[\alpha=\left\lceil\frac{84\ln(1/\varepsilon)}{1-486c}\right\rceil\]
Then there exists an instance of the $(3\alpha)$-center problem such that, in the distributed setting with $9$ processors, each communicating at most a $c$ fraction of its input points to the coordinator, the coordinator fails to output a better than $4$-approximation with probability at least $1-\varepsilon$.
\end{lemma}

\begin{proof}
The underlying metric space consists of $\alpha$ disjoint copies of $\mathcal{M}(n')$ separated by an arbitrarily large distance from one another. 
The point set of each copy is distributed to the nine processors as described earlier, and these distribtions are independent. Thus, each processor receives $\alpha\cdot(n'+3)$ points. Observation~\ref{obs_lb_opt} implies that in this instance, the optimum set of $3\alpha$ centers (the union of optimum sets of $3$ centers in each copy) has unit cost. Also, in order to get a better than $4$-approximation, the coordinator must output a better than $4$-approximate solution from every copy. We prove that this is unlikely.

By our assumption, each processor sends at most $c\alpha\cdot(n'+3)$ points to the coordinator, where $c<1/486$. Therefore, for each processor, there exist at most $54c\alpha$ copies from which it sends more than $(n'+3)/54$ points to the coordinator. Since we have $9$ processors, there exist at most $9\times54c\alpha=486c\alpha$ copies from which more than $(n'+3)/54$ points are sent by some processor. From each of the remaining $(1-486c)\alpha$ copies, no processor sends more than $(n'+3)/54$ points. By Lemma~\ref{thm:lb}, the coordinator succeeds on each of these copies independently with probability at most $1-1/84$, in producing a better than $4$ approximation. Therefore, the probability that the coordinator succeeds in all the $(1-486c)\alpha$ copies is bounded as
\[\left(1-\frac{1}{84}\right)^{(1-486c)\alpha}\leq\exp\left(-\frac{1-486c}{84}\cdot\alpha\right)\leq\varepsilon\text{,}\]
where the last inequality follows by substituting the value of $\alpha$. Thus, the coordinator fails to produce a better than $4$-approximation with probability at least $1-\varepsilon$.
\end{proof}

\begin{proof}[\textbf{Proof of Lemma~\ref{thm:lb}}]
Consider any one of the nine processors. It gets the set $\pi(Y_i^j)$ for a uniformly random $(i,j)\in[3]\times[3]$. Since $\pi$ is a uniformly random labeling and points in $Y_i^j$ are pairwise equidistant, the processor is not able to identify the three critical points in its input. This happens even if we condition on the values of $\Gamma$. Formally,
conditioned on $\Gamma$ and $\mathcal{P}$,
all subsets of $\mathcal{P}_r$ of size $3$ are equally likely to be the set of labels of the three critical points in processor $r$'s input, i.e., $Y_i^j$ where $(i,j)=\Gamma^{-1}(r)$.
As a consequence, the probability that at least one of the three critical points appears in the set of at most $\ell$ points the processor communicates is at most $3\ell/|Y_i^j|=3\ell/(n'+3)$, even when we condition on $\Gamma$.
For a given processor $r\in[9]$, let $O_r$ be the set of labels it sends to the coordinator, and define $B_r$ to be the event that $O_r$ contains the label of a critical point. Then $\Pr[B_r\mid\Gamma,\mathcal{P}]\leq3\ell/(n'+3)$. Next, define $G$ to be the event that no processor sends the label of any critical point to the coordinator, that is, $G = \cap_{r=1}^9 B_r^c$, where $B_r^c$ is the complement of $B_r$. Then by the union bound and the fact that $\ell\leq (n'+3)/54$, we have for every partition $P$ of the label set and every bijection $\gamma:[3]\times[3]\longrightarrow[9]$,
\begin{equation}\label{eqn_G}
\Pr[G\mid\Gamma=\gamma,\mathcal{P}=P] \geq 1-9\cdot\frac{3\ell}{n'+3}\geq \frac{1}{2}\text{.}
\end{equation}


Suppose the coordinator outputs $O$, a set of three labels, on receiving $O_1,\ldots,O_9$. Then $O\subseteq O_{r_1}\cup O_{r_2}\cup O_{r_3}$ for some $r_1,r_2,r_3\in[9]$. Observe that $O_1,\ldots,O_9$, $O$, and $\{r_1,r_2,r_3\}$ are all completely determined\footnote{If $O$ intersects less than three of the $O_r$'s, then we define $\{r_1,r_2,r_3\}$ to be the lexicographically smallest set such that $O\subseteq O_{r_1}\cup O_{r_2}\cup O_{r_3}$.} by $\mathcal{P}$. In contrast, due to the random labeling $\pi$, the mapping $\Gamma$ is independent of $\mathcal{P}$. Therefore,
\begin{observation}\label{obs_gamma_info}
Conditioned on $\mathcal{P}$, the bijection $\Gamma$ is equally likely to be any of the $9!$ bijections from $[3]\times[3]$ to $[9]$.
\end{observation}
Next, define $G'$ to be the event that $\{r_1,r_2,r_3\}$ is either $\Gamma(\{(1,1),(1,2),(1,3)\})$ or $\Gamma(\{(3,1),(3,2),(3,3)\})$. In words, $G'$ is the event that the coordinator outputs labels of three points, all of which are contained
in $Y_1^1\cup Y_1^2\cup Y_1^3$ or in $Y_3^1\cup Y_3^2\cup Y_3^3$.
Note that the event $G'\cap G$ implies that the coordinator's output is contained
in $S_1^1\cup S_1^2\cup S_1^3=S_1$ or in $S_3^1\cup S_3^2\cup S_3^3=S_3$. Therefore, by Observation \ref{obs_lb_opt}, event $G'\cap G$ implies that the coordinator fails to output a better than $4$-approximation. We are now left to bound $\Pr[G'\cap G]$ from below.

Since the set $\{r_1,r_2,r_3\}$ is completely determined by $\mathcal{P}$, the event $G'$ is completely determined by $\mathcal{P}$ and $\Gamma$: for any $\mathcal{P}$, there exist exactly $2\cdot3!\cdot6!$ values of $\Gamma$ which cause $G'$ to happen. Formally,
\begin{observation}\label{obs_G'}
For every partition $P$ of the label set, there exist exactly $2\cdot3!\cdot6!$ bijections $\gamma:[3]\times[3]\longrightarrow[9]$ such that $\Pr[G'\mid\mathcal{P}=P,\Gamma=\gamma]=1$, whereas $\Pr[G'\mid\mathcal{P}=P,\Gamma=\gamma']=0$ for all the other bijections $\gamma':[3]\times[3]\longrightarrow[9]$.
\end{observation}

Therefore, we have,
\begin{eqnarray*}
\Pr[G\cap G'] & = & \sum_{P,\gamma}\Pr[G\cap G'\mid\mathcal{P}=P,\Gamma=\gamma]\cdot\Pr[\mathcal{P}=P,\Gamma=\gamma]\\
 & = & \sum_{(P,\gamma):\Pr[G'\mid\mathcal{P}=P,\Gamma=\gamma]=1}\Pr[G\mid\mathcal{P}=P,\Gamma=\gamma]\cdot\Pr[\Gamma=\gamma\mid\mathcal{P}=P]\cdot\Pr[\mathcal{P}=P]\\
 & \geq & \sum_P\sum_{\gamma:\Pr[G'\mid\mathcal{P}=P,\Gamma=\gamma]=1}\frac{1}{2}\cdot\frac{1}{9!}\cdot\Pr[\mathcal{P}=P]\\
 & = & \frac{1}{2}\cdot\frac{1}{9!}\cdot\sum_P\left|\{\gamma:\Pr[G'\mid\mathcal{P}=P,\Gamma=\gamma]=1\}\right|\cdot\Pr[\mathcal{P}=P]\\
 & = & \frac{2\cdot3!\cdot6!}{2\cdot9!}\cdot\sum_P\Pr[\mathcal{P}=P]\\
 & = & \frac{1}{84}\text{.}
\end{eqnarray*}

Here, we used Observation \ref{obs_G'} for the second and fourth equality, and Equation (\ref{eqn_G}) and Observation \ref{obs_gamma_info} for the inequality. Thus, the coordinator fails to output a better than $4$-approximation with probability at least $1/84$, as required.
\end{proof}

Using Lemma~\ref{lem:amplifiedlb} along with Yao's lemma, we get our main lower-bound theorem.
\begin{theorem}\label{thm:yao}
There exists $c > 0$ such that for any $\varepsilon > 0$, with $k = \Theta(\log (1/\varepsilon))$, any randomized distributed algorithm for $k$-center where each processor communicates at most $cn$ points to the coordinator, who outputs a subset of those points as the solution, is no better than $4$-approximation with probability at least $1 - \varepsilon$.
\end{theorem}


%% file: experiments1.tex
\section{Experiments}
All experiments are run on HP EliteBook 840 G6 with Intel\textsuperscript{\textregistered} Core\texttrademark~i7-8565U
CPU \@ 1.80GHz having 4 cores and 15.5 GiB of RAM, running Ubuntu 18.04 and Anaconda.  We make our code available on GitHub\footnote{\url{https://github.com/sagark4/fair_k_center}}.

We perform our experiments on a massive synthetic dataset, several real datasets, and small synthetic datasets.
The same implementation is used for the large synthetic dataset and the real datasets, but a slightly different implementation is used for small synthetic datasets.
Before presenting the experiments, we first discuss the implementation details that are common to all three experiments.
Specific details are mentioned along with the corresponding experimental setup.
For all our algorithms if the solution size is less than $k$, then we extend the
solution using an arbitrary solution of size $k$ (which also certifies the simple upper bound).
In the case of the distributed algorithm, an arbitrary solution is computed using only the points received by the coordinator.  Also, one extra pass is spent into computing solution cost.
In the processors' algorithm, we return $r_i$ along with $(P_i,L_i)$.
No randomness is used for any optimization, making our algorithms completely deterministic.  Access to distance between two points is via a method \texttt{get\_distance()}, whose implementation
depends on the dataset.  

We use the code shared by Kleindessner et al.\ for their algorithm on github\footnote{\url{https://github.com/matthklein/fair_k_center_clustering}}, exactly as is, for all datasets.  In their code, the distance is assumed to be stored in an $n \times n$ distance matrix.

As mentioned in the introduction, we give new implementations for
existing algorithms---those of Chen et al.\ and Kale (we choose to implement Kale's two-pass algorithm only, because it is the better of his two).  
Instead of using a matroid intersection subroutine, which can have running time of super quadratic in $n$,
we reduce the postprocessing steps of these algorithms to finding a maximum
matching in an appropriately constructed graph (for details, see \texttt{HittingSet()} in Appendix~\ref{app:alg}).  We further reduce maximum matching to max-flow which is computed using Python package NetworkX.
This results in a postprocessing time of
$O(k^2n)$ for Chen et al.~and $O(k^3)$ for Kale.  
This step itself makes Chen et al.'s algorithm practical for much larger $n$
than what is observed by Kleindessner et al.

\paragraph*{Handling the guesses}
For all algorithms (except Kleindessner et al.'s), we use $\eps = 0.1$.  For Chen et al.'s algorithm, we use geometric guessing starting with the lower bound given by the farthest point heuristic
(call this Gonzalez's lower bound).
For our two-pass algorithm and Kale's algorithm, we use geometric guessing starting with the simple lower bound until the upper bound given by an arbitrary solution.
The values for the guesses $\tau$ in the coordinator's algorithm are scaled down by a factor of $5.1$. 
Concretely, let $r_1$ be the maximum among the $r_i$'s. 
Then the guesses take values in $\frac{r_1}{5.1},\frac{1.1\cdot r_1}{5.1},\frac{(1.1)^2\cdot r_1}{5.1},\ldots$, until a feasible solution is found.
The factor of $5.1$ ensures that when \texttt{getPivots()} is run with the parameter $10\tau < 2r_1$, we end up picking at least $k$ pivots from $X'$.

We now proceed to present our experiments. 
To show the effectiveness of our algorithms on massive datasets, we run them on
a 100 GB synthetic dataset which is a collection of 4,000,000 points in 1000
dimensional Euclidean space, where each coordinate is a uniformly random real in
$(0, 10000)$.  
Each point is assigned one of the four groups uniformly at random,
and capacity of each group is set to $2$.
Just reading this data file takes more than four minutes.
Our two-pass algorithm takes 1.95
hours and our distributed
algorithm takes 1.07 hours; both compute a solution of almost the same cost, even though their theoretical guarantees are different.  Here, we use block size of $10000$ in the distributed algorithm, i.e., the number of
processors $\ell = 400$.

\paragraph*{For the above dataset and the real datasets:} The input is read from the input file and attributes are
read from the attribute file, one data point at a time, and fed to the
algorithms.  This is done in order to be able to handle the 100 GB dataset.  Using Python's multiprocessing library, we are able to use
four cores of the
processor~\footnote{\url{https://www.praetorian.com/blog/multi-core-and-distributed-programming-in-python}}.

\subsection{Real Datasets}
We use three real world
datasets: Celeb-A~\cite{liu_etal15}, Sushi~\cite{sushi_dataset}, and
Adult~\cite{adults_dataset}, with $n = 1000$ by selecting the first 1000 data
points (see Table~\ref{tab:real}).
\begin{table*}[t]
  \caption{Comparison of solution quality of algorithms for fair $k$-center on
    real datasets.  Each column after the third corresponds to an algorithm and
    shows ratio of its cost and Gonzalez's lower bound.  Note that this is not
    the approximation ratio.  Our two-pass algorithm is the best for majority of
    the settings. Dark shaded cell shows the best-cost algorithm and lightly shaded cell shows the second best.}
  \label{tab:real}
  \centering
    \begin{tabular}{|l|l|p{2cm}|l|l|p{2cm}|l|l|} 
      \hline
      Dataset & Capacities & Gonzalez's Lower Bound & Chen et al. & Kale & Kleindessner et al. & Two pass & Distributed \\ \hline
      CelebA & [2, 2] & \hfill 30142.4 & 1.9 & 1.9 & \cellcolor{gray!25!white}1.85 & \cellcolor{gray!80!white}1.76 & \cellcolor{gray!80!white}1.76 \\ \hline
      CelebA & [2, 2, 2, 2] & \hfill 28247.3 & 2.0 & 2.0 & \cellcolor{gray!25!white}1.9 & \cellcolor{gray!80!white}1.88 & \cellcolor{gray!80!white}1.88 \\ \hline
      SushiA & [2, 2] & \hfill 11.0 & 2.18 & 2.18 & 2.27 & \cellcolor{gray!80!white}2.0 & \cellcolor{gray!40!white}2.09 \\ \hline
      SushiA & [2] * 6 & \hfill 8.5 & \cellcolor{gray!25!white}2.35 & \cellcolor{gray!25!white}2.35 & \cellcolor{gray!80!white}2.24 & \cellcolor{gray!25!white}2.35 & \cellcolor{gray!80!white}2.24 \\ \hline
      SushiA & [2] * 12 & \hfill 7.5 & \cellcolor{gray!25!white}2.13 & \cellcolor{gray!25!white}2.13 & \cellcolor{gray!80!white}2.0 & 2.4 & 2.4 \\ \hline
      SushiB & [2, 2] & \hfill 36.5 & \cellcolor{gray!80!white}1.81 & \cellcolor{gray!80!white}1.81 & 2.11 & \cellcolor{gray!80!white}1.81 & \cellcolor{gray!25!white}1.86 \\ \hline
      SushiB & [2] * 6 & \hfill 34.0 & 2.0 & \cellcolor{gray!25!white}1.82 & 2.12 & \cellcolor{gray!80!white}1.79 & 2.0 \\ \hline
      SushiB & [2] * 12 & \hfill 32.0 & \cellcolor{gray!80!white}1.94 & \cellcolor{gray!80!white}1.94 & \cellcolor{gray!25!white}2.09 & \cellcolor{gray!80!white}1.94 & \cellcolor{gray!80!white}1.94 \\ \hline
      Adult & [2, 2] & \hfill 4.9 & 2.04 & 2.13 & 2.44 &\cellcolor{gray!80!white} 1.9 & \cellcolor{gray!25!white}2.02 \\ \hline
      Adult & [2] * 5 & \hfill 3.92 & 2.66 & 2.66 & \cellcolor{gray!80!white}2.02 & 2.36 & \cellcolor{gray!25!white}2.35 \\ \hline
      Adult & [2] * 10 & \hfill 2.76 & 2.75 & \cellcolor{gray!80!white}2.41 & \cellcolor{gray!25!white}2.48 & \cellcolor{gray!25!white}2.48 & 2.75 \\ \hline
    \end{tabular}

\end{table*}

Celeb-A dataset is a set of 202,599 images of human faces with attributes
including male/female and young/not-young, which we use.  We use Keras to
extract features from each image~\cite{feature_extraction} via the pretrained
neural network VGG16, which returns a 15360 dimensional real vector for each
image.  We use the $\ell_1$ distance as the metric and two settings of groups:
male/female with capacity of $2$ each (denoted by $[2,2]$ in
Table~\ref{tab:real}), and \{male, female\} $\times$ \{young, not-young\} with
capacity of $2$ each (denoted by $[2] * 4$ in Table~\ref{tab:real}).

Sushi dataset is about preferences for different types of Sushis by 5000
individuals with attributes of male/female and six possible age-groups.  In
SushiB, the preference is given by a score whereas in SushiA, the preference is
given by an order.  For SushiB, we use the $\ell_1$ distance whereas for SushiA, we
use the number of inversions, i.e., the distance between two Sushi rankings is
the number of doubletons $\{i, j\}$ such that Sushi $i$ is preferred over Sushi
$j$ by one ranking and not the other.  For both SushiA and SushiB, we use three
different group settings: with gender only, with age group only, and combination
of gender and age group.  This results in $2$, $6$, and $12$ groups,
respectively, and the capacities appear as $[2, 2]$, $[2] * 6$, and $[2] * 12$,
respectively, in Table~\ref{tab:real}.

Motivated by Kleindessner et al., we consider the adult dataset~\cite{adults_dataset}, which is extracted from US census data and
contains male/female attribute and six numerical attributes that we use as
features.  We normalize this dataset to have zero mean and standard deviation of
one and use the $\ell_1$ distance as the metric.  There are two attributes that can
be used to generate groups: gender and race (Black, White, Asian Pacific
Islander, American Indian Eskimo, and Other).  Individually and in combination,
this results in $2$, $5$, and $10$ groups, respectively.

For comparison, see Table~\ref{tab:real}.  On majority of
settings, our two-pass algorithm outputs a solution with cost smaller than the
rest.  We reiterate for emphasis that in addition to being at least as good as
the best in terms of solution quality, our algorithms can handle massive
datasets.

For the distributed algorithm, we use block size of 25, i.e., the number of
processors are $1000/25 = 40$: theoretically, using $\approx\sqrt{n}$ processor gives maximum speedup.


%% file: experiments2.tex
\subsection{Synthetic Datasets}
Motivated by the experiments in Kleindessner et al.,
we use the Erd\H{o}s-R{\'e}nyi  graph metric to compare the running time and cost of our algorithms with existing algorithms. 
For a fixed natural number $n$, a random metric on $n$ points is generated as follows. First, a random undirected graph on $n$ vertices is sampled in which each edge is independently picked with probability $2\log n/n$. 
Second, every edge is assigned a uniformly random weight in $(0,1000)$. 
The points in the metric correspond to the vertices of the graph, and the pairwise distances between the points are given by the shortest path distance. 
In addition, if $m$ is the number of groups, then each point in the metric is assigned a group in $\{1,2,\ldots,m\}$ uniformly and independently at random.

\begin{figure}[h]
\centering
  \includegraphics[scale=0.5]{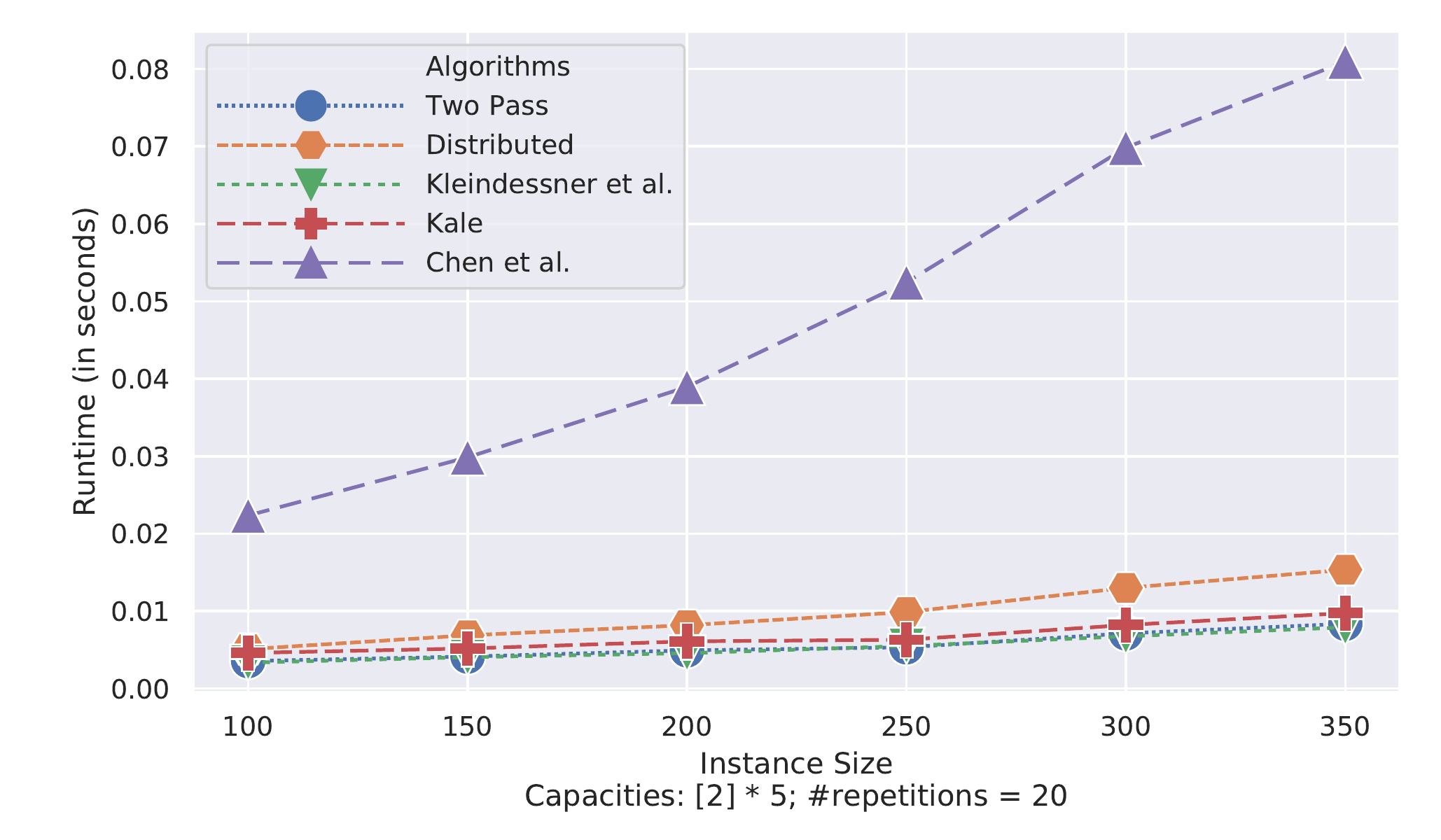}
  \includegraphics[scale=0.55]{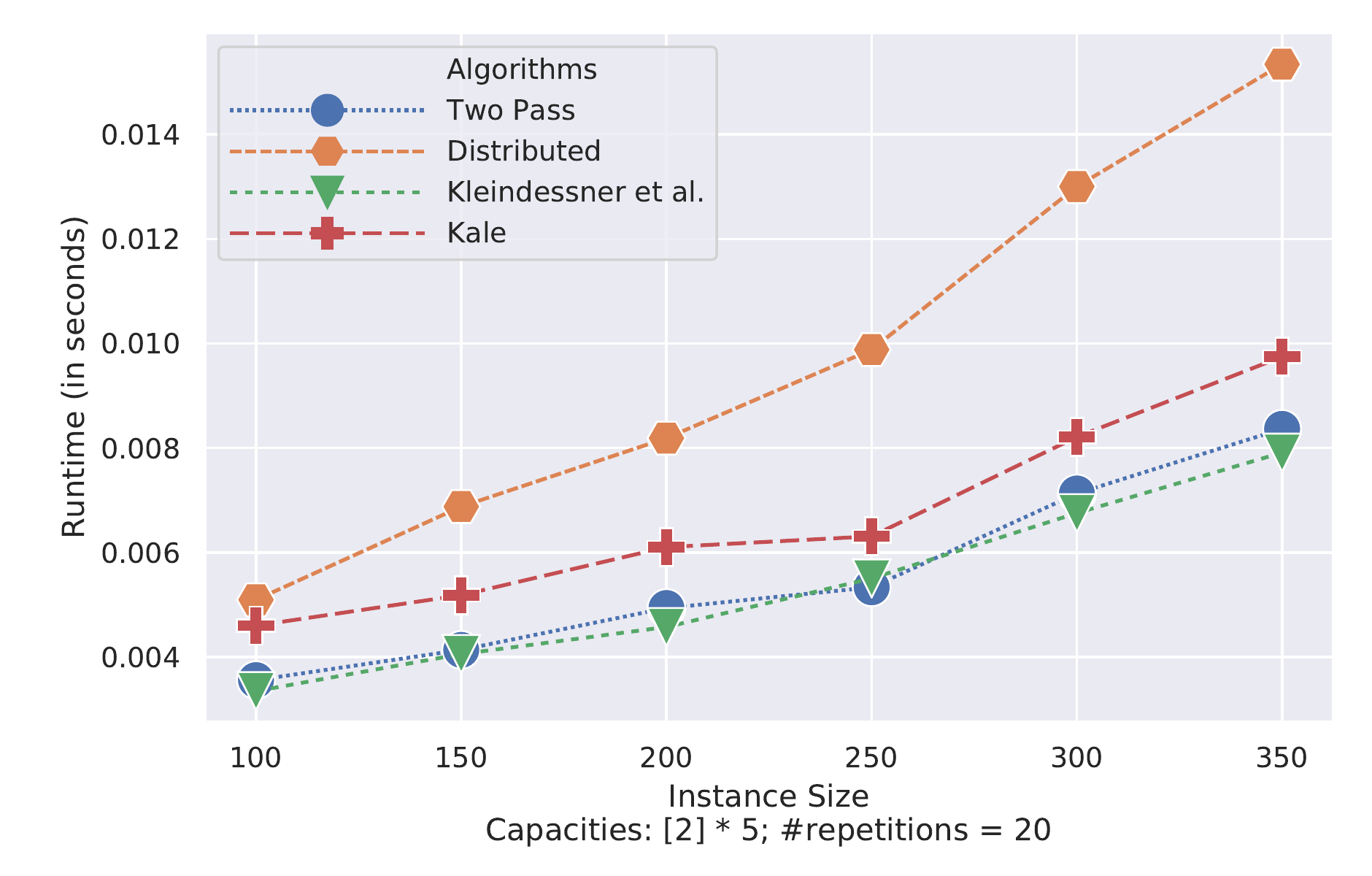}
  \vspace{-0.5cm}
  \caption{Comparing Running Times}
  \label{fig:runtimes}
\end{figure}
Figure \ref{fig:runtimes} shows the plots between the running time and instance size $n$; the bottom one is a zoom-in of the top one to the lower four plots.
In this experiment, $n$ takes values in $\{100,150,200,\ldots,350\}$.
The number of groups is fixed to $5$ and the capacity of each group is $2$.
For each fixing of $n$, we run the five algorithms on $20$ independent random metric instances of size $n$ to compute the average running time.
Our two pass algorithm and Kleindessner et al.'s algorithm are the fastest.
Our distributed algorithm is faster than Chen et al.'s algorithm, but slower than Kale's.
\begin{figure}[h]
  \centering
  \includegraphics[trim={0.0cm 0.5cm 0 0.73cm},clip,scale=0.65]{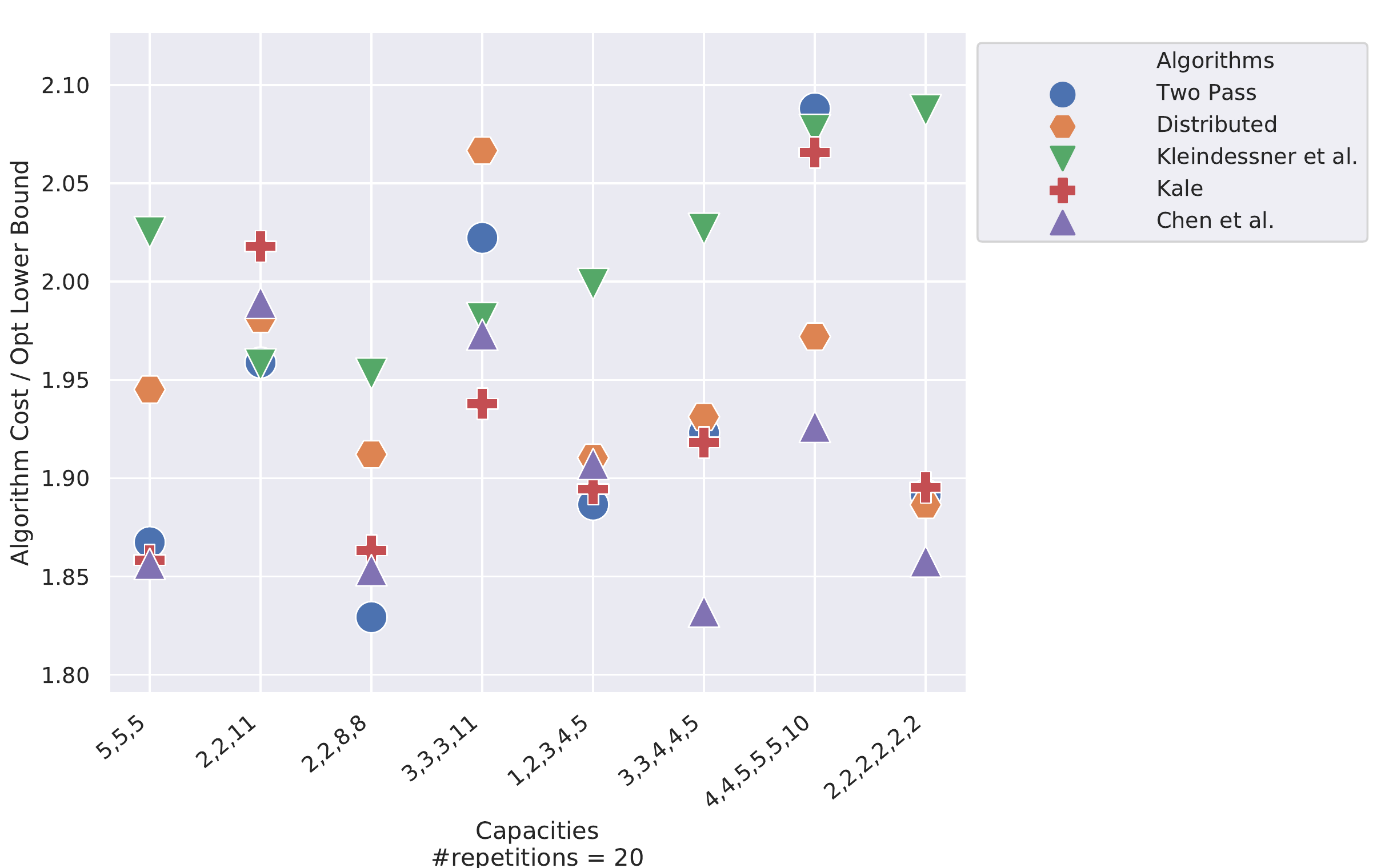}
  \caption{Comparing Approximation Ratios}
  \label{fig:approx}
\end{figure}

Figure \ref{fig:approx} shows the ratios of the cost of various algorithms to Gonzalez's lower bound.
For this comparison, the instance size is fixed to $500$ and capacities are $[5, 5, 5], [2, 2, 11], [2, 2, 8, 8], [3, 3, 3, 11], [1, 2, 3, 4, 5]$, $[3, 3, 4, 4, 5], [4, 4, 5, 5, 5, 10], [2, 2, 2, 2, 2, 2].$
Here again, for every fixing of the capacities, the algorithm is run on $20$ independent random metric instances to compute the average costs.
Chen et al.'s algorithm achieves the least cost for almost all settings, and Kleindessner et al.'s algorithm gives the highest cost on majority (5 out of 8) of settings.
Our two-pass algorithm and Kale's algorithm perform similar to each other and are quite close to Chen et al.'s.  Our distributed algorithm is somewhere in between Chen et al.'s and Kleindessner et al.'s.
Note that the ratios of the costs between any two algorithms is at most $1.167$.

In the implementation of our two pass algorithm, we use geometric guessing starting with the simple lower bound until the algorithm returns a success instead of running all guesses.  This is done for a fair comparison in terms of running time.

%% file: conclusion.tex
\section{Research Directions}

One research direction is to improve
the theoretical bounds, e.g., get a better approximation ratio in the
distributed setting or prove a better hardness result.  Another interesting
direction is to use fair $k$-center for fair rank aggregation using the number
of inversions between two rankings as the metric.


%% file: supplementary.tex
\section{Algorithms}
\label{app:alg}
The definition of clustering cost (Definition 1) immediately implies the following observations.

\begin{observation}\label{obs_inclusion}
Let $A\supseteq A'$ and $B\subseteq B'$ be sets of points in a metric space given by a distance function $d$. The clustering cost of $A$ for $B$ is at most the clustering cost of $A'$ for $B'$.
\end{observation}

\begin{observation}\label{obs_union}
Let $A_1,A_2,B_1,B_2$ be sets of points in a metric space given by a distance function $d$. Suppose the clustering cost of each $A_i$ for $B_i$ is at most $\tau$. Then the clustering cost of $A_1\cup A_2$ for $B_1,\cup B_2$ is at most $\tau$.
\end{observation}

The following lemma follows easily from the triangle inequality.

\begin{lemma}[Lemma 1 from the paper, restated]
Let $A,B,C\subseteq X$. The clustering cost of $A$ for $C$ is at most the clustering cost of $A$ for $B$ plus the clustering cost of $B$ for $C$.
\end{lemma}

\begin{proof}
Let $d$ be the metric and let $r_{AB}$ and $r_{BC}$ denote the clustering costs of $A$ for $B$ and of $B$ for $C$ respectively. For every $a\in A$, there exists $b\in B$ such that $d(a,b)\leq r_{AB}$. But for this $b$, there exists $c\in C$ such that $d(b,c)\leq r_{BC}$. Thus, for every $a\in A$, there exists a $c\in C$ such that $d(a,c)\leq r_{AB}+r_{BC}$, by the triangle inequality. This proves the claim.
\end{proof}


The pseudocodes of procedures \texttt{getPivots}$()$, \texttt{getReps}$()$, and \texttt{HittingSet}$()$ are given by Algorithms~\ref{alg_pivots},~\ref{alg_representatives}, and~\ref{alg_bmatching} respectively.

\begin{algorithm}[htb]
\caption{\texttt{getPivots}$(T,d,r)$}
\label{alg_pivots}
\begin{algorithmic}
\State {\bfseries Input:} Set $T$ with metric $d$, radius $r$.
\State $P\gets\{p\}$ where $p$ is an arbitrary point in $T$.
\For{{\bfseries each} $q\in T$ (in an arbitrary order)}
    \If{$\min_{p\in P}d(p,q)>r$}
        \State $P\gets P\cup\{q\}$.
    \EndIf
\EndFor
\State {\bfseries Return} $P$.
\end{algorithmic}
\end{algorithm}

\begin{algorithm}[htb]
\caption{\texttt{getReps}$(T,d,g,P,r)$}
\label{alg_representatives}
\begin{algorithmic}
\State {\bfseries Input:} Set $T$ with metric $d$, group assignment function $g$, subset $P\subseteq T$, radius $r$.
\For{{\bfseries each} $p\in P$}
    \State $I_p\gets\{p\}$.
\EndFor
\For{{\bfseries each} $q\in T$ (in an arbitrary order)}
    \For{{\bfseries each} $p\in P$}
        \If{$d(p,q)\leq r$ and $I_p$ doesn't contain a point from $q$'s group}
            \State $I_p\gets I_p\cup\{q\}$.
        \EndIf
    \EndFor
\EndFor
\State {\bfseries Return} $\{I_p:p\in P\}$.
\end{algorithmic}
\end{algorithm}

\begin{algorithm}[htb]
\caption{\texttt{HittingSet}$(\mathcal{N},g,\overline{k})$}
\label{alg_bmatching}
\begin{algorithmic}
\State {\bfseries Input:} Collection $\mathcal{N}=(N_1,\ldots,N_K)$ of pairwise disjoint sets of points, group assignment function $g$, vector $\overline{k}=(k_1,\ldots,k_m)$ of capacities.
\State Construct bipartite graph $G=(\mathcal{N},V,E)$ as follows.
\State $V$ $\gets$ $\biguplus_{j=1}^mV_i$, where $V_j$ is a set of $k_j$ vertices.
\For{{\bfseries each} $N_i$ and {\bfseries each} group $j$}
    \If{$\exists$ $p\in N_i$ such that $g(p)=j$}
        \State Connect $N_i$ to all vertices in $V_j$.
    \EndIf
\EndFor
\State Find the maximum cardinality matching $H$ of $G$.
\State $C\gets\emptyset$.
\For{{\bfseries each} edge $(N_i,v)$ of $H$}
    \State Let $p$ be a point in $N_i$ from group $j$, where $v\in V_j$.
    \State $C\gets C\cup\{p\}$.
\EndFor
\State {\bfseries Return} $C$.
\end{algorithmic}
\end{algorithm}

\begin{observation}\label{obs_getPivots}
The procedure \texttt{getPivots}$()$ performs a single pass over the input set $T$. The set $P$ returned by \texttt{getPivots}$(T,d,r)$ contains points separated pairwise by distance more than $r$. The clustering cost of $P$ for $T$ is at most $r$. Therefore, by Lemma 2 from the paper, if there is a set of $k$ points whose clustering cost for $T$ is at most $r/2$, then $|P|\leq k$ pivots.
\end{observation}

\begin{observation}\label{obs_representatives}
The procedure \texttt{getRep}$()$ executes a single pass over the input set $T$. The points in each set $I_p$ returned by \texttt{getRep}$(T,d,g,P,r)$ belong to distinct groups and are all within distance $r$ from $p$. For every point $q$ within distance $r$ from $p\in P$, $I_p$ contains a point in the same group as $q$ (possibly $q$ itself).
\end{observation}

The procedure \texttt{HittingSet}$()$ constructs the following bipartite graph. The left side vertex set contains $K$ vertices: one for each $N_i$. The right side vertex set is $V=\biguplus_{j=1}^mV_j$, where $V_j$ contains $k_j$ vertices for each group $j$. If $N_i$ contains a point from group $j$, then its vertex is connected to the all of $V_j$. Each matching $H$ in this bipartite graph encodes a feasible subset $C$ of $\biguplus_{i=1}^KN_i$ as follows. For each edge $e=(N_i,v)\in H$ where $v\in V_j$, add to $C$ the point from $N_i$ belonging to group $j$. Observe that since $|V_j|=k_j$ and $H$ is a matching, $C$ contains at most $k_j$ points from group $j$. Moreover, $|C|=|H|$, and hence, a maximum cardinality matching in the bipartite graph encodes a set $C$ intersecting as many of the $N_i$'s as possible.

In our implementation, we enhance the efficienty of \texttt{HittingSet}$()$ as follows. For each group, we introduce only one vertex in the right side vertex set and construct the bipartite graph like \texttt{HittingSet}$()$, directing edges from left to right. We further connect a source to the left side vertices with unit capacity edges, and the right side vertices to a sink with edges of capacities $k_j$. We find the maximum (integral) source-to-sink flow using the Ford-Fulkerson algorithm. 
For each $i$ and $j$, if the edge $(N_i,j)$ exists and carries nonzero flow, then we include in $C$ the point in $N_i$ that belongs to group $j$. Our runtime is bounded as follows.

\begin{lemma}
The runtime of \texttt{HittingSet}$()$ is $O(K^2\cdot\max_i|N_i|)$. 
\end{lemma}

\begin{proof}
The number of edges in the constructed bipartite graph is $O(K\cdot\max_i|N_i|)$ whereas the value of the max-flow is no more than $K$. The runtime of the Ford-Fulkerson algorithm is of the order of the size of the number of edges times the value of max-flow. Therefore, the runtime of \texttt{HittingSet}$()$, which is dominated by the runtime of the Ford-Fulkerson algorithm, turns out to be $O(K^2\cdot\max_i|N_i|)$.
\end{proof}

%
